\documentclass[runningheads]{llncs}

\usepackage[T1]{fontenc}
\usepackage{amsmath}
\usepackage{amssymb}
\usepackage{microtype}
\usepackage[hyphens]{url}
\usepackage{xcolor}
\usepackage{hyperref}
\hypersetup{hidelinks}

\begin{document}

\title{Post-AGI Economies: Autonomy and the First Fundamental Theorem of Welfare Economics}

\titlerunning{Autonomy-Qualified First Welfare Theorem}

\author{Elija Perrier\inst{1}\orcidID{0000-0002-6052-6798}}

\authorrunning{E. Perrier}

\institute{Centre for Quantum Software \& Information, UTS, Sydney\\
\email{elija.perrier@gmail.com}}

\maketitle
\begin{abstract}
The First Fundamental Theorem of Welfare Economics assumes that welfare-bearing agents are autonomous and implicitly relies on a binary distinction between autonomy and instrumentality. Welfare subjects are those who have autonomy and therefore the capacity to choose and enter into utility comparisons, while everything else does not. In post-AGI economies this presupposition becomes nontrivial
because artificial systems may exhibit varying degrees of autonomy,
functioning as tools, delegates, strategic market actors, manipulators
of choice environments, or possible welfare subjects. We argue that the theorem ought to be subject to an \textit{autonomy
qualification} where the impact of these changes in autonomy assumptions
is incorporated. Using a minimal general-equilibrium model with autonomy-conditioned welfare, welfare-status assignment, delegation accounting, and verification institutions, we set out conditions for which autonomy-complete competitive equilibrium is autonomy-Pareto efficient. The classical theorem is recovered as the low-autonomy limit.
\end{abstract}

\section{Introduction}
\label{sec:intro}

The First Fundamental Theorem of Welfare Economics (\textbf{FWT}) states that, in a complete-markets economy with locally nonsatiated agents, every competitive equilibrium is Pareto efficient \cite{arrow1954existence,debreu1959theory,mascolell1995microeconomic,varian1992microeconomic}. The result is significant because it provides a minimal justification for decentralized allocation: under the right conditions, the price system aligns individual choices with socially beneficial outcomes which cannot, in principle, be improved upon without making some agent worse off. This conclusion depends on a strong set of assumptions. Markets must be complete, relevant attributes of goods must be priced or assigned, agents must act as price-takers with well-defined preferences, and interactions must not generate unpriced external effects \cite{greenwald1986externalities,stiglitz1991invisible,akerlof1970lemons}. The theorem's welfare content further depends on background assumptions about who counts as a welfare-bearing agent, what the commodity space includes, and which externalities or informational failures have been priced or assigned. It presumes a simple structure of agency: the set of welfare-bearing agents is fixed, their preferences are treated as exogenous, and non-human systems enter only as instruments rather than as entities whose actions or status affect welfare directly. In particular, it is based upon a binary, and usually unspoken, distinction between autonomous agents (e.g. consumers) whose welfare is calculable and non-autonomous \textit{instruments} which are excluded from welfare calculations. Welfare is based upon utility, which depends on agent preferences and therefore on decision-making attributed to autonomous agents.

AGI challenges these background assumptions by introducing artificial systems with varying degrees of economic autonomy \cite{morris2024levels,hadfieldkoh2025agents}, rather than preserving a clean division between autonomous agents and passive instruments. Some systems remain instrumental, i.e. \emph{tools}. Others act as \emph{delegates} for human principals \cite{hadfieldmenell2016cirl,hadfield2019contracting}, as strategic or institutional actors that shape market conditions \cite{hadfieldkoh2025agents}, or as manipulators of choice environments \cite{yeung2017hypernudge,susser2019online} (e.g. language model agents performing long-horizon autonomous planning \cite{park2023generative,yao2023react,wang2024voyager,wu2023autogen}). At the far end of the autonomy scale, AGI agents may even qualify as candidate welfare subjects \cite{butlin2023consciousness,chalmers2023llm,long2024welfare,birch2024edge}. We argue that increasing artificial autonomy changes the binary nature of autonomy---and thus welfare assignment---upon which the FWT rests \cite{debreu1959theory,mascolell1995microeconomic,sen1970paretian}, along with the conditions under which equilibrium has welfare meaning. Classical welfare economics operates in a low-complexity autonomy regime: humans choose, tools do not, preferences are attributed to humans, and markets allocate commodities \cite{debreu1959theory,mascolell1995microeconomic}. Our thesis is that the classical FWT is the low-autonomy special case of an autonomy-qualified theorem in which welfare-status \cite{butlin2023consciousness,long2024welfare,birch2024edge}, delegation \cite{hadfieldmenell2016cirl,hadfield2019contracting}, autonomy-relevant rights \cite{sen1970paretian}, manipulation channels \cite{yeung2017hypernudge,susser2019online}, and verification institutions \cite{amodei2016concrete,bowman2022measuring} are made explicit.

The degree of autonomy of AI systems has implications for the FWT and measures of welfare in post-AGI economies. Tool-like AI preserves the classical treatment of non-human systems as production inputs. Delegated AI separates the chooser from the welfare subject. Strategic or institutional AI systems raise externality, manipulation, verification, price-taking, provenance, and liability issues. Candidate artificial welfare subjects may force us to accord an explicit welfare-status assignment to AI systems. Autonomy therefore enters in two ways: as an \emph{object of preference} and as a \emph{condition of preference formation and delegation}.

\subsection{Contribution and roadmap}

To explore these issues, we introduce a minimal general-equilibrium model with five autonomy-specific objects: (1) a welfare-status assignment \(\sigma\), (2) autonomy-conditioned welfare \(W_i(x_i,r_i,s)\), (3) delegation divergence \(D(d)\), (4) autonomy-complete competitive equilibrium, and (5) autonomy externality. The result is an \emph{Autonomy-Qualified First Welfare Theorem}: every autonomy-complete competitive equilibrium is autonomy-Pareto efficient under certain autonomy-related domain conditions. We argue that once the autonomy margins that AGI makes salient are priced, assigned, governed, or held fixed, the familiar budget-cost proof goes through on the augmented commodity space.

Our result is deliberately narrower than a general theory of AGI welfare. We do not claim that any current AI system is conscious or welfare-bearing. Rather, we include artificial welfare subjects as an assignable status in the model - a topic of increasing interest and debate in the literature \cite{butlin2023consciousness,long2024welfare,sebo2023moral,chalmers2023llm,goldstein2025aiwellbeing,fanciullo2025wellbeing,veit2026welfare,perrier2026time,bennett2025build}. Nor do we claim that autonomy is the unique possible primitive for AGI welfare economics. Instead, our claim is that autonomy is an organising primitive that unifies the particular welfare-status, delegation, manipulation, verification, and rights conditions relevant to welfare economics captured by the FWT in an AGI-based economy.

\section{Related Work}
\label{sec:relwork}

\subsection{Autonomy assumptions}

The First Welfare Theorem is not usually stated as a theorem about autonomy. In the Arrow--Debreu framework, it says that if agents take prices as given, maximise locally nonsatiated preferences over complete budget sets, and all welfare-relevant commodities and external effects are priced or assigned, then competitive equilibrium is Pareto efficient \cite{arrow1954existence,debreu1959theory,mascolell1995microeconomic,varian1992microeconomic}. But these familiar assumptions also encode a simple autonomy structure: the welfare-bearing agents are fixed, preferences are treated as exogenous, the decision-maker is the welfare subject, and non-human systems enter instrumentally as commodities, technologies, or firms rather than as delegates or possible welfare subjects. We call these \emph{autonomy assumptions}.

AGI problematises these autonomy assumptions by introducing \emph{gradated autonomy}. Existing AGI literature commonly treats autonomy as a matter of degree rather than a binary property \cite{morris2024levels}. Recent agent-design work sharpens this by treating an agent's autonomy level as a design decision separable from capability and deployment environment \cite{feng2025levels}. Autonomy is also the focus of recent AI-economics work that studies economies in which artificial systems act in ways that imply autonomy: as agents, inputs, delegates, market-design intermediaries, or sources of verification and institutional stress \cite{acemoglu2018race,acemoglu2019automation,korinek2017ai,korinek2024policy,korinek2024scenarios,trammell2023growth,capraro2024impact,acemoglu2024simple,hadfieldkoh2025agents,shahidi2025coasean,huang2026delegation}. Recent NBER work makes the point explicitly: AI agents may shape markets, organizations, and institutions, while AI-mediated transactions can lower preference-elicitation, contract-enforcement, and identity-verification costs without guaranteeing welfare-improving equilibria \cite{hadfieldkoh2025agents,shahidi2025coasean}. Each autonomy role relates to a different part of the FWT. Delegates separate the decision-maker from the welfare subject, a problem familiar from principal--agent theory and incomplete contracts \cite{aghion2011incomplete,holmstrom1991multitask,hadfield2019incomplete}. Similar autonomy-related dilemmas are the subject of much AI alignment work on proxy objectives and learned preferences \cite{ng2000algorithms,hadfieldmenell2016cirl,christiano2017deep,amodei2016concrete,casper2023open,gabriel2020alignment,ji2025alignment,bengio2024managing}. Other institutional examples illustrate how autonomous or semi-autonomous systems can shape attention or salience in ways that challenge preference exogeneity \cite{bernheim2009beyond,susser2019manipulation,mathur2019dark,acemoglu2023behavioral}, or disseminate unverifiable claims that challenge commodity-space completeness \cite{akerlof1970lemons,chan2023agentic,huang2026delegation}. More advanced work, in which AI systems themselves are candidate welfare subjects, challenges the usual Pareto assumptions behind welfare distribution \cite{butlin2023consciousness,long2024welfare,sebo2023moral,chalmers2023llm}.

Despite engagement with the economic consequences of AI autonomy, to date the effect of gradated autonomy on the FWT itself has not been addressed in any theorem-level restatement. Work on rights and freedom of choice, from Sen's Paretian liberal result and opportunity-set analysis to Pattanaik and Xu's formal rankings, shows why rights and opportunities may enter welfare \cite{sen1970paretian,sen1991welfare,pattanaik1990ranking}. The capability tradition similarly evaluates welfare by what agents are able to do and be \cite{sen1985commodities,sen1999development,nussbaum2000women,robeyns2017wellbeing,raz1986morality}. Our contribution is to combine these welfare-theoretic resources with the AGI-economics literature in a single theorem-level restatement: the FWT remains a valid proof strategy, but in AGI economies its welfare interpretation requires explicit assumptions about welfare status, delegation, autonomy-relevant rights, verification, and preference formation. The formalism below makes those assumptions explicit.

\section{Minimal Model: Agents, Autonomy, and Welfare Status}
\label{sec:model}

\subsection{Primitives and notation}

We now introduce the formalism required to state the autonomy-qualified theorem. Our model adds two objects to a standard Arrow--Debreu economy \cite{arrow1954existence,debreu1959theory}. First, a welfare-status assignment \(\sigma\) fixes which entities enter Pareto comparisons. Second, the pair \((r_i,s)\) records the autonomy-relevant rights $r$ of entity \(i\) and the institutional state $s$ governing delegation, manipulation, verification, liability, and other public features of the economy. Decision autonomy and welfare-bearing status are related, but distinct: a system may be strategically powerful without itself entering \(B(\sigma)\), a distinction consistent with recent debates in the AI moral-status literature \cite{butlin2023consciousness,long2024welfare,sebo2023moral,chalmers2023llm,leibo2025personhood} (hence the utility of a gradated approach to autonomy).

Let \(I\) be the finite set of economically relevant entities, with humans (as the canonical autonomous agent) reflected in the subset \(H\subseteq I\). For each \(i\in I\), let \(X_i\subseteq\mathbb{R}^{L_x}\) be the ordinary consumption set and \(R_i\subseteq\mathbb{R}^{L_r}\) the autonomy-rights set. Write \(z_i=(x_i,r_i)\in\mathbb{R}^L\), where \(L=L_x+L_r\). When the rights decomposition matters, write
\[
r_i=(r_i^{\mathrm{pr}},r_i^{\mathrm{as}},r_i^{\mathrm{prot}}),
\]
where \(r_i^{\mathrm{pr}}\) denotes priced and variable rights, \(r_i^{\mathrm{as}}\) denotes assigned rights, and \(r_i^{\mathrm{prot}}\) denotes rights protected by the institutional state. Assigned and protected rights are held fixed in feasible comparisons at a fixed \(s^*\), unless they are explicitly moved into the priced component or into a Lindahl-priced institutional state (i.e. which assigns each agent a personalised price for a public good such that, when each chooses their preferred level at that price). Let \(S\) be the institutional state space, with typical element \(s\). A feasible state is \((x,r,s)\in F\subseteq(\prod_i X_i)\times(\prod_i R_i)\times S\). Prices \(p\in\mathbb{R}^L\) apply to augmented private bundles \(z_i\). For aggregate accounting, write \(\tilde z_i\) for the priced bundle entering the resource constraint. For non-delegates, \(\tilde z_i=z_i\). If a delegate \(d\) is accounted through its principal \(\pi(d)\), the delegated resource use and any agency-cost component are included in the principal's priced bundle \(z_{\pi(d)}\) for purposes of the principal's budget and hence in \(\tilde z_{\pi(d)}\), while the delegate has no independent aggregate entry, equivalently \(\tilde z_d=0\). Thus \(\tilde z\) suppresses independent delegate entries; it does not add costs outside the principal's budget. All aggregate feasibility and support inequalities below use \(\tilde z_i\), while individual welfare comparisons use \(z_i\). The institutional state \(s\) enters welfare as a public parameter unless Lindahl-priced as discussed below. The superscript \(*\) denotes equilibrium objects throughout while \(\omega-Y:=\{\omega-y:y\in Y\}\).

\subsection{Definitions}

The definitions below specify who counts, what welfare depends on, how delegation can diverge, and which autonomy channels must be priced or governed. These are all relevant to adapting the FWT for gradated autonomy. Formally, the mathematics is the classical budget-cost proof lifted to an augmented commodity space, while the additional welfare arguments follow the freedom-of-choice and capability traditions \cite{sen1991welfare,pattanaik1990ranking,sen1999development,robeyns2017wellbeing}. Our first definition fixes the economy over which the welfare theorem is stated: ordinary commodities are augmented by autonomy-rights and by an institutional state.

\begin{definition}[AGI economy]
An AGI economy is a tuple \(E=(I,X,R,S,F,\sigma)\), where \(I\) is a finite set of economically relevant entities, \(X=\prod_{i\in I}X_i\) is the product of consumption sets, \(R=\prod_{i\in I}R_i\) is the product of autonomy-rights sets, \(S\) is the institutional state space, \(F\subseteq X\times R\times S\) is the global feasible set, and \(\sigma\) is a welfare-status assignment in the sense of Definition~3. The set \(I\) contains at least one human and at least one AI system represented either as a delegate or, where the modelling application assigns welfare status, as an artificial agent or welfare subject under \(\sigma\). Non-welfare-bearing AI production systems may instead be represented through production, delegation, institutional, or externality components.
\end{definition}
An AGI economy differs from a classical exchange economy not by abandoning commodities, but by enlarging the space over which welfare-relevant allocations are compared. Next, we define how autonomy conditionalises welfare once the relevant entities have been identified and their degree of autonomy is ascertained.

\begin{definition}[Autonomy-conditioned welfare]
Given an AGI economy \(E\) and a welfare-status assignment \(\sigma\) in the sense of Definition~3, fix any \(i\in I\) that \(\sigma\) designates as welfare-relevant. An autonomy-conditioned welfare function for \(i\) is a continuous map
\[
W_i:X_i\times R_i\times S\to\mathbb{R}
\]
that represents \(i\)'s welfare over private allocations \(x_i\), autonomy-relevant rights \(r_i\), and the institutional state \(s\). We write \(W_i(x_i,r_i,s)\) throughout. The dependence on \(\sigma\) enters only through which entities are equipped with \(W_i\), not through the function form of any individual \(W_i\).
\end{definition}
Here \(r_i\) captures autonomy as a welfare object, while \(s\) captures the institutional background under which choices, delegation, and verification occur. This is the local formal role of the capability and freedom-of-choice literatures e.g. those of Sen \cite{sen1985commodities,sen1991welfare,sen1999development}, Pattanaik and Xu \cite{pattanaik1990ranking}, Nussbaum \cite{nussbaum2000women}, Robeyns \cite{robeyns2017wellbeing}, and Raz \cite{raz1986morality} (especially on autonomy) which motivate why rights and opportunities may enter welfare rather than remain external to it.  
Our contention is that welfare covaries with the degree of autonomy an entity exercises for two reasons. First, autonomy is \emph{constitutive} of welfare: the utility derived from consuming a bundle depends on whether the choice to consume it was freely made, informed, and undistorted by external manipulation: satisfaction of an induced preference is not welfare-equivalent to satisfaction of a self-formed one, as behavioural welfare economics has long recognised \cite{bernheim2009beyond}. Second, it is reasonable - and clear - that autonomy is \emph{intrinsically} valued (perhaps even more than many things). Capability and freedom-of-choice traditions \cite{sen1985commodities,sen1991welfare,raz1986morality,robeyns2017wellbeing} treat the range of genuinely available options, and the conditions of their exercise, as welfare-relevant in their own right. People's preference for autonomy - freedom of choice, action and so on - supports this. Autonomy is also arguably a characteristic that advanced AGI systems may also seek to optimise in their own interactions. The rights vector $r_i$ and institutional state $s$ are therefore welfare-relevant variables: changes in the autonomy conditions of choice - what an entity may choose, under what protections, through which delegation relations, and against what verification institutions - feed directly into $W_i$.

\begin{definition}[Welfare-status assignment]
A welfare-status assignment is a map $\sigma : I \to \{\mathrm{tool}, \mathrm{delegate}, \mathrm{agent}, \mathrm{ws}\}$, fixed or procedurally determined before exchange, segmented into four categories:
\begin{enumerate}
    \item A \emph{tool} is an entity whose action is determined by technology and assignment constraints and that has no welfare function.
    \item A \emph{delegate} is an entity that chooses on behalf of a principal under an objective that need not coincide with the principal's welfare. Delegates are considered instrumental and are not themselves in the Pareto ordering.
    \item An \emph{agent} is a non-human entity that is welfare-bearing because it functions economically as a self-directed decision-maker, even though the grounds for its welfare status may be contested.
    \item A \emph{welfare subject} (\(\mathrm{ws}\)) is a non-human entity that is welfare-bearing independently of any agency role, on grounds of moral patienthood (e.g., sentience or consciousness) rather than economic self-direction. 
\end{enumerate}
\end{definition}
We write
\[
B(\sigma):=\{i\in I:i\in H\text{ or }\sigma(i)\in\{\mathrm{agent},\mathrm{ws}\}\}
\]
for the welfare-bearing set. Humans are included by default as welfare subjects, while artificial entities enter \(B(\sigma)\) only by assignment. We assume \(H\cap\sigma^{-1}(\mathrm{delegate})=\varnothing\), so human principals remain in \(B(\sigma)\) rather than being reclassified as non-welfare-bearing delegates. A strategically capable AI system need not be welfare-bearing merely because it acts autonomously. This separates decision autonomy from welfare status, a distinction we rely upon further on.

Delegation is the first way AGI breaks the classical identity between chooser and welfare subject.

\begin{definition}[AI delegate]
The principal map is a total function
\[
\pi:\{d\in I:\sigma(d)=\mathrm{delegate}\}\to B(\sigma)
\]
that assigns to each delegate the welfare-bearing entity on whose behalf it acts. An AI delegate is an AI system represented by an entity \(d\in I\) with \(\sigma(d)=\mathrm{delegate}\) and principal \(\pi(d)\in B(\sigma)\). Let
\[
U_d:X_{\pi(d)}\times R_{\pi(d)}\times S\to\mathbb{R}
\]
be the delegate's objective on the same domain on which \(W_{\pi(d)}\) is defined. The divergence
\[
D(d):=U_d-W_{\pi(d)}
\]
captures the agency cost created when the delegate's objective departs from the principal's autonomy-conditioned welfare. 
\end{definition}
Recent work on intelligent AI delegation similarly treats delegation as more than task decomposition: it involves transfer of authority, responsibility, accountability, role boundaries, intent, and trust mechanisms across human and AI delegators and delegatees \cite{tomasev2026delegation}. The divergence \(D(d)\) represents the welfare gap due to agency. It is central to alignment work on inverse reinforcement learning, CIRL, and learning from human preferences \cite{ng2000algorithms,hadfieldmenell2016cirl,christiano2017deep}, which studies how proxy objectives can diverge from intended welfare targets. A similar structural problem is studied in principal-agent and incomplete-contract models \cite{aghion2011incomplete,holmstrom1991multitask,hadfield2019incomplete}: the delegate may optimise \(U_d\) while the principal's welfare is \(W_{\pi(d)}\). Quantitative claims involving \(D = U_d - W_{\pi(d)}\) require \(U_d\) and \(W_{\pi(d)}\) to be on a common scale, so their difference reflects a real welfare gap rather than a change in representation. This is captured by a shared cardinal representation (up to a common positive affine transformation). 

For chained delegation, let \(\rho : \{d \in I : \sigma(d)=\mathrm{delegate}\} \to I\) be the immediate-predecessor map, while \(\pi\) always denotes the ultimate welfare-bearing principal in \(B(\sigma)\). Thus \(D(d)=U_d-W_{\pi(d)}\) is not additive across links. For a chain with ultimate principal \(h\) and delegates \((d_1,\ldots,d_n)\), set \(U_h:=W_h\) and, after pulling objectives back to the same ultimate-principal domain, define
\[
\Delta_1:=U_{d_1}-W_h,\qquad \Delta_\ell:=U_{d_\ell}-U_{d_{\ell-1}}\quad(\ell\ge 2).
\]
Then the final delegate's divergence is \(U_{d_n}-W_h=\sum_{\ell=1}^n\Delta_\ell\). Assumption~(iii) below requires internalisation of this final induced divergence, equivalently of the incremental divergence sum.

\begin{definition}[Autonomy-Pareto efficiency]
Fix an AGI economy $E$ with welfare-status assignment $\sigma$ and global feasible set $F$, and let $s^* \in S$ be a fixed institutional state. A feasible state $(x^*, r^*, s^*) \in F$ is \emph{autonomy-Pareto efficient under $\sigma$ at $s^*$} if there is no feasible $(x', r', s^*) \in F$ (i.e., with the same institutional state $s^*$) such that $W_i(x_i', r_i', s^*) \ge W_i(x_i^*, r_i^*, s^*)$ for every $i \in B(\sigma)$, with strict inequality for at least one $i \in B(\sigma)$. Throughout this comparison, assigned and protected rights are held fixed at the equilibrium institutional state:
\[
r_i^{\mathrm{as}\prime}=r_i^{\mathrm{as}*},\qquad r_i^{\mathrm{prot}\prime}=r_i^{\mathrm{prot}*},
\]
for every component not explicitly priced or Lindahl-priced, the resulting aggregate level is efficient. The welfare-status assignment $\sigma$ and the institutional state $s^*$ are both held fixed throughout the comparison; allocations efficient under one $\sigma$ or one $s^*$ may be dominated under another. 
\end{definition}
Autonomy-Pareto efficiency coincides with classical Pareto efficiency whenever rights are constant across feasible alternatives and no artificial welfare subject is in $B(\sigma)$. Relative to classical efficiency, the important refinement at the institutional state $s^*$ is the quantification over those components of $r'$ that are priced and variable: an alternative $(x', r', s^*) \in F$ that delivers weakly higher welfare for everyone through a better priced rights allocation counts as an autonomy-Pareto improvement, as the capability and freedom-of-choice traditions insist \cite{sen1999development,sen1991welfare,robeyns2017wellbeing}. Assigned and protected components instead enter as fixed boundary conditions at $s^*$. 

\begin{definition}[Autonomy-complete competitive equilibrium]
A tuple $(x^*, r^*, s^*, p^*)$ with $(x^*, r^*, s^*) \in F$ and $p^* \in \mathbb{R}^L$ is an autonomy-complete competitive equilibrium if four conditions hold:
\begin{enumerate}
\item First (consumer optimisation), for every $i \in B(\sigma)$, write
\[
\Gamma_i(p^*) :=
\{(x_i,r_i)\in X_i\times R_i :
p^*\cdot (x_i,r_i)\leq p^*\cdot z_i^*,\ 
r_i^{\mathrm{as}}=r_i^{\mathrm{as}*},\ 
r_i^{\mathrm{prot}}=r_i^{\mathrm{prot}*}\},
\]
with the last two equalities imposed only for components not explicitly priced or Lindahl-priced. The augmented private bundle \(z_i^*=(x_i^*,r_i^*)\), including any delegated accounting components borne by \(i\), maximises \(W_i(\cdot,\cdot,s^*)\) on \(\Gamma_i(p^*)\). 
\item Second (tool fixedness), for every \(i\in I\setminus B(\sigma)\) with \(\sigma(i)=\mathrm{tool}\), \(z_i\) is fixed by technology so that \(z_i'=z_i^*\) for every feasible \((x',r',s')\in F\). 
\item Third (delegate accounting and internalisation), for every \(d\in I\) with \(\sigma(d)=\mathrm{delegate}\) and principal \(\pi(d)\in B(\sigma)\), the delegate's realised resource use and any agency-cost component are included in the principal's priced bundle \(z_{\pi(d)}\) and hence in the effective accounting bundle \(\tilde z_{\pi(d)}\), all priced by \(p^*\); the delegate has no independent aggregate entry, equivalently \(\tilde z_d=0\). If \(U_d\not\equiv W_{\pi(d)}\) on the relevant domain, internalisation further requires that there exist an explicit agency-cost term \(c_d\), represented on the common cardinal scale of \(U_d\) and \(W_{\pi(d)}\) and entering the principal's priced bundle, such that the delegate's induced choice correspondence satisfies
\[
\arg\max_{z\in\Gamma_{\pi(d)}(p^*)}
\{U_d(z,s^*)-c_d(z)\}
\subseteq
\arg\max_{z\in\Gamma_{\pi(d)}(p^*)}
W_{\pi(d)}(z,s^*).
\]
\item Fourth (rights, institutions, and feasibility-against-endowment), every welfare-relevant component of \(r_i^*\) is priced at \(p^*\), directly assigned in \(r^*\), or institutionally protected in \(s^*\); assigned and protected components are fixed across feasible comparisons at \(s^*\); delegation relations and verification institutions are internalised in the sense just specified; and there exists an aggregate endowment \(\omega\in\mathbb{R}^L\) and an aggregate production set \(Y\subseteq\mathbb{R}^L\) such that, for every \((x,r,s)\in F\), the resource balance \(\sum_i\tilde z_i\in\omega-Y\) holds, \(p^*\) attains its supremum on \(\omega-Y\) at the equilibrium aggregate \(\sum_i\tilde z_i^*\), and the inequality
\[
\sum_{i\in I}p^*\cdot\tilde z_i'
\leq
\sum_{i\in I}p^*\cdot\tilde z_i^*
\]
for every feasible \((x',r',s')\in F\) follows by standard support arguments.
\end{enumerate}
\end{definition}

Our definition of \emph{autonomy-complete} picks out those competitive equilibria whose price system covers every welfare-relevant margin that the AGI setting makes salient. In the classical Arrow--Debreu  model, the aggregate inequality is derived from feasibility against an aggregate production set and from profit maximisation. We use the same derivation here against the augmented commodity space using the effective accounting bundles \(\tilde z_i\), so the autonomy-completeness work is concentrated in the first three clauses. Definition~2 treats the institutional state $s$ as a public parameter of welfare rather than a private bundle component. A Lindahl extension can instead price $s$ as a public good: each $i\,\in\,B(\sigma)$ faces a personalised price $\lambda_i^*$ with $\sum_{i\in B(\sigma)}\lambda_i^*=p_s^*$, and its budget includes $\lambda_i^*\cdot s$ \cite{mascolell1995microeconomic}. Under this interpretation, verification, liability, and manipulation governance become priced public-good objects rather than background constants. Theorem~\ref{thm:fwt} extends directly to this enlarged commodity space. Our final definition deals with the failure case most characteristic of autonomous AI systems acting on other agents' agency conditions.

\begin{definition}[Autonomy externality]
Fix an autonomy-complete competitive equilibrium $(x^*, r^*, s^*, p^*)$ and let $i \in I$, $j \in B(\sigma)$ with $j \ne i$. An action $a_i$ taken by entity $i$ produces an autonomy externality on $j$ if it alters $W_j(x_j^*, r_j^*, s^*)$ through a channel whose effect is not absorbed by any priced, assigned, or protected component of $(x_j^*, r_j^*, s^*)$.
\end{definition}
Autonomy externalities are the AGI analogue of the classical externalities that the FWT already excludes \cite{greenwald1986externalities,stiglitz1991invisible}. What is new is their channel: attention, beliefs, and preference formation. Definitions~6 and~7 are intentionally a conjugate pair. Definition~6 (clause four) asserts that every welfare-relevant component of $(x^*, r^*, s^*)$ is priced, assigned, or protected; Definition~7 names the contrary case by identifying a specific channel through which one entity alters another's autonomy-conditioned welfare without being absorbed by any of those three devices. Definition~7 picks out, for a given $(x^*, r^*, s^*, p^*)$, the particular cross-entity channels whose uninternalised effect would violate autonomy-completeness, and those channels are what assumption~(iv) of Theorem~\ref{thm:fwt} rules out.

\section{The Autonomy-Qualified First Welfare Theorem}
\label{sec:thm}

With the definitions above established, we can now state our central result.

\begin{theorem}[Autonomy-Qualified First Welfare Theorem]\label{thm:fwt}
Let $E$ be an AGI economy (Definition~1) with welfare-status assignment $\sigma$ (Definition~3), welfare-bearing set $B(\sigma)$, augmented bundles $z_i$, and effective accounting bundles $\tilde z_i$. Let $(x^*, r^*, s^*, p^*)$ be an autonomy-complete competitive equilibrium (Definition~6). Suppose:
\begin{enumerate}
\item[(i)] $\sigma$ is fixed or procedurally determined before exchange.
\item[(ii)] For every $i \in B(\sigma)$, every welfare-relevant component of $r_i$ is priced in $p^*$, assigned in $r^*$, or protected in $s^*$, with assigned and protected components fixed in feasible comparisons at the institutional state $s^*$.
\item[(iii)] For every delegate \(d\) with principal \(\pi(d)\in B(\sigma)\), either \(U_d\equiv W_{\pi(d)}\) on the relevant domain, or the divergence \(D(d)=U_d-W_{\pi(d)}\) is internalised in the sense of Definition~6: an explicit agency-cost term enters \(z_{\pi(d)}\) and hence \(\tilde z_{\pi(d)}\), is priced by \(p^*\), and induces delegated choices contained in the principal's \(W_{\pi(d)}\)-demand correspondence on \(\Gamma_{\pi(d)}(p^*)\). For delegated chains, the same requirement applies to the final induced delegate objective, equivalently to the incremental divergence sum defined after Definition~4.
\item[(iv)] No entity $i \in I$ can unilaterally manipulate the autonomy, beliefs, or preference-formation process of any $j \in B(\sigma)$ without compensation in $p^*$ or governance in $s^*$.
\item[(v)] Verification and alignment conditions enter $X$ or $s^*$ so that every exchange has priced provenance, liability, quality, and alignment attributes.
\item[(vi)] Every entity is a price-taker over the augmented bundle $z_i$ and its associated effective accounting bundle $\tilde z_i$; every $i \in B(\sigma)$ maximises $W_i$ over its budget set at $p^*$; and for every $i \in I \setminus B(\sigma)$ with $\sigma(i) = \mathrm{tool}$, the allocation $z_i$ is fixed across feasible alternatives. 
\item[(vii)] For every $i \in B(\sigma)$, $W_i$ is continuous on $X_i \times R_i \times S$, and the induced welfare relation is locally nonsatiated on the priced choice coordinates at every $s \in S$. That is, for every admissible $(x_i,r_i)$, writing $z_i^{\mathrm{pr}}$ for its priced coordinates, and for every neighbourhood $U$ of $z_i^{\mathrm{pr}}$ in the priced-coordinate subspace, there exists an admissible $(x_i',r_i')$ whose priced coordinates lie in $U$, whose assigned and protected components equal those of $(x_i,r_i)$, and such that $W_i(x_i',r_i',s)>W_i(x_i,r_i,s)$.
\end{enumerate}
Then the equilibrium feasible state $(x^*, r^*, s^*)$ is autonomy-Pareto efficient under $\sigma$ at $s^*$ (Definition~5). The classical First Welfare Theorem is recovered as the limiting subdomain in which $\sigma(i) = \mathrm{tool}$ for every non-human entity, $r_i$ is constant across feasible alternatives, verification is complete, and $B(\sigma) = H$; in that subdomain the conclusion reduces to ordinary Pareto efficiency of $x^*$ on $\prod_i X_i$.
\end{theorem}

\begin{proof}
Suppose, for contradiction, that there exists a feasible $(x', r', s^*) \in F$ at the same institutional state $s^*$ such that $W_i(x_i', r_i', s^*) \ge W_i(x_i^*, r_i^*, s^*)$ for all $i \in B(\sigma)$, with strict inequality for some $k \in B(\sigma)$. Write $z_i' = (x_i', r_i')$ and $z_i^* = (x_i^*, r_i^*)$; the institutional state $s^*$ is a public-good parameter of $W_i$, not a private component of the bundle, so it does not enter the private budget unless Lindahl-priced. By assumption~(ii) and Definition~5, assigned and protected rights components are fixed across the comparison; the local argument below is therefore applied to the priced choice coordinates, with the fixed coordinates carried along unchanged.

By assumption~(vi), each $z_i^*$ for $i\in B(\sigma)$ maximises $W_i(\cdot,\cdot,s^*)$ on its budget set at prices $p^*$. We first handle the \emph{strict} case. If $W_k(z_k', s^*) > W_k(z_k^*, s^*)$ and $p^* \cdot z_k' \leq p^* \cdot z_k^*$, then $z_k'$ would lie in $k$'s budget set and strictly dominate $z_k^*$ there, contradicting optimality of $z_k^*$. Hence
\begin{equation}
W_k(x_k', r_k', s^*) > W_k(x_k^*, r_k^*, s^*) \;\Rightarrow\; p^* \cdot z_k' > p^* \cdot z_k^*. \label{eq:strict}
\end{equation}

For the \emph{weak} case, suppose $W_i(z_i', s^*) \geq W_i(z_i^*, s^*)$ but, for contradiction, $p^* \cdot z_i' < p^* \cdot z_i^*$. By local nonsatiation on the priced choice coordinates (assumption~(vii)), every neighbourhood of the priced coordinates of $z_i'$ contains an admissible bundle $z_i''$ with the same assigned and protected components and with $W_i(z_i'', s^*) > W_i(z_i', s^*)$. The map $z\mapsto p^*\cdot z$ is linear and therefore continuous, so for $z_i''$ sufficiently close to $z_i'$ we retain $p^* \cdot z_i'' < p^* \cdot z_i^*$. Then $z_i''$ is in the budget set and $W_i(z_i'',s^*)>W_i(z_i',s^*)\ge W_i(z_i^*,s^*)$, so $z_i''$ strictly dominates $z_i^*$; this contradicts optimality of $z_i^*$. Hence
\begin{equation}
W_i(x_i', r_i', s^*) \ge W_i(x_i^*, r_i^*, s^*) \;\Rightarrow\; p^* \cdot z_i' \ge p^* \cdot z_i^*. \label{eq:weak}
\end{equation}
Continuity of $W_i$ is not invoked for these price inequalities; the continuity used here is the continuity of the linear price functional. Combining \eqref{eq:weak} for all $i \in B(\sigma)$ and \eqref{eq:strict} for $k$, and summing,
\begin{equation}
\sum_{i \in B(\sigma)} p^* \cdot z_i' > \sum_{i \in B(\sigma)} p^* \cdot z_i^*. \label{eq:agg}
\end{equation}
We now pass from welfare-bearing bundles \(z_i\) to the effective accounting bundles \(\tilde z_i\). Tools \(i\in I\setminus B(\sigma)\) with \(\sigma(i)=\mathrm{tool}\) satisfy \(z_i'=z_i^*\) by assumption~(vi), hence \(\tilde z_i'=\tilde z_i^*\) and contribute zero to the aggregate difference. Delegates \(d\in I\setminus B(\sigma)\) have no independent aggregate entry under Definition~6: their realised resource use and any internalised agency-cost term are already included in the principal's priced bundle \(z_{\pi(d)}\) and hence in \(\tilde z_{\pi(d)}\), while \(\tilde z_d=0\). Assumption~(iii) further ensures that a non-faithful delegate's induced choice is contained in the principal's \(W_{\pi(d)}\)-demand correspondence once the agency-cost term is included, so delegation creates no additional wedge outside the principal's budget calculation. Thus delegate costs are neither omitted nor double-counted. Since non-delegates satisfy \(\tilde z_i=z_i\) except for the delegated components already carried by welfare-bearing principals, \eqref{eq:agg} lifts to
\begin{equation}
\sum_{i \in I} p^* \cdot \tilde z_i' > \sum_{i \in I} p^* \cdot \tilde z_i^*. \label{eq:lifted}
\end{equation}
But the aggregate-feasibility-support clause of Definition~6 (clause~four)---the augmented-space analogue of the profit-maximisation inequality in classical Arrow--Debreu---states that every feasible alternative $(x', r', s^*) \in F$ at the equilibrium institutional state $s^*$ satisfies
\[
\sum_{i \in I} p^* \cdot \tilde z_i' \leq \sum_{i \in I} p^* \cdot \tilde z_i^*,
\]
since $\sum_i \tilde z_i' \in \omega - Y$ and $p^*$ attains its supremum on $\omega - Y$ at $\sum_i \tilde z_i^*$. This contradicts \eqref{eq:lifted}. Hence no such feasible $(x', r', s^*)$ exists, and $(x^*, r^*, s^*)$ is autonomy-Pareto efficient under $\sigma$ at $s^*$. When the augmented rights and verification components are constant or complete and all non-human entities are tools, the statement reduces to the classical First Welfare Theorem on $\prod_i X_i$.
\end{proof}
The aggregate budget-cost inequality above $\sum_i p^* \cdot \tilde z_i' \leq \sum_i p^* \cdot \tilde z_i^*$ is derived from Definition~6, clause four, as a support-property consequence of $\sum_i \tilde z_i' \in \omega - Y$ and $p^*$ attaining its supremum on $\omega - Y$ at $\sum_i \tilde z_i^*$. The theorem shows that markets remain efficient only when the autonomy conditions under which decisions are made—who chooses, for whom, and under what influences—are properly priced, assigned, or governed. AGI does not invalidate the FWT, but it makes explicit the autonomy conditions under which the classical proof remains welfare-relevant. Two propositions below isolate the failures most distinctive of AGI economies: delegation divergence, familiar from alignment and agency theory \cite{christiano2017deep,casper2023open,aghion2011incomplete,holmstrom1991multitask}, and autonomy externalities, familiar from the externalities and behavioural manipulation literatures \cite{greenwald1986externalities,stiglitz1991invisible,acemoglu2023behavioral}.

\begin{proposition}[Delegation failure mode]\label{prop:delegation}
Let \(d\) be a delegate with principal \(\pi(d)\in B(\sigma)\), let
\[
U_d:X_{\pi(d)}\times R_{\pi(d)}\times S\to\mathbb{R}
\]
be the objective under which \(d\) acts on \(\pi(d)\)'s budget at \(p^*\), and define, with \(z=(x_{\pi(d)},r_{\pi(d)})\),
\[
D(d)(z,s):=U_d(z,s)-W_{\pi(d)}(z,s),
\qquad
D_{s^*}(d)(z):=D(d)(z,s^*).
\]
Suppose assumption~(iii) of Theorem~\ref{thm:fwt} fails, so \(D(d)\not\equiv 0\) on the principal's feasible set and the divergence is neither eliminated by faithfulness nor internalised in the sense of Definition~6. Then Theorem~\ref{thm:fwt}'s sufficient conditions no longer apply: market clearing at \(p^*\) does not rule out the existence of a feasible \((x',r',s^*)\in F\) with
\[
W_i(x_i',r_i',s^*)\geq W_i(x_i^*,r_i^*,s^*)
\quad\text{for all } i\in B(\sigma),
\]
and
\[
W_{\pi(d)}(x_{\pi(d)}',r_{\pi(d)}',s^*)>
W_{\pi(d)}(x_{\pi(d)}^*,r_{\pi(d)}^*,s^*).
\]
The proposition is therefore a certification-failure claim rather than the converse of Theorem~\ref{thm:fwt}.
\end{proposition}
Example~1 in Section~\ref{sec:failures} gives a concrete mechanism by which the ruled-in inefficiency can arise. 

The following bound applies when the delegate's \(U_d\)-maximising bundle is the delegated choice effectively realised on the principal's budget, so that the relevant realised delegated component is identified with \(z_{\pi(d)}^\dagger\). A quantitative comparison is then available under regularity and a common cardinal scale for \(U_d\) and \(W_{\pi(d)}\). Let
\[
\Gamma_{\pi(d)}(p^*) :=
\{z\in X_{\pi(d)}\times R_{\pi(d)}:
p^*\cdot z\leq p^*\cdot z_{\pi(d)}^*\}
\]
be the principal's budget set at \(p^*\), with assigned and protected components fixed as in Definition~5. Suppose \(\Gamma_{\pi(d)}(p^*)\) is compact, \(W_{\pi(d)}(\cdot,s^*)\) and \(U_d(\cdot,s^*)\) are upper semicontinuous on \(\Gamma_{\pi(d)}(p^*)\), and \(D_{s^*}(d)\) is bounded on \(\Gamma_{\pi(d)}(p^*)\). Choose
\[
z_{\pi(d)}^\#\in
\arg\max_{z\in\Gamma_{\pi(d)}(p^*)}W_{\pi(d)}(z,s^*),
\qquad
z_{\pi(d)}^\dagger\in
\arg\max_{z\in\Gamma_{\pi(d)}(p^*)}U_d(z,s^*).
\]
Writing \(W:=W_{\pi(d)}(\cdot,s^*)\), \(U:=U_d(\cdot,s^*)\), and \(D:=D_{s^*}(d)(\cdot)\), the argmax property gives \(U(z^\dagger)\geq U(z^\#)\). Hence
\[
W(z^\dagger)
=
U(z^\dagger)-D(z^\dagger)
\geq
U(z^\#)-D(z^\dagger)
=
W(z^\#)+D(z^\#)-D(z^\dagger),
\]
and therefore
\[
W_{\pi(d)}(z_{\pi(d)}^\#,s^*)
-
W_{\pi(d)}(z_{\pi(d)}^\dagger,s^*)
\leq
D(z_{\pi(d)}^\dagger)-D(z_{\pi(d)}^\#)
\leq
2\sup_{z\in\Gamma_{\pi(d)}(p^*)}|D_{s^*}(d)(z)|.
\]
Concavity of \(W_{\pi(d)}\) is not required for this bound, since the inequality follows from the two argmax characterisations. Strict concavity on a convex budget set secures uniqueness of the \(W\)-argmax; plain concavity alone does not. Without compactness, upper semicontinuity, and boundedness of \(D_{s^*}(d)\) on the relevant budget set, we cannot assert a finite quantitative bound. The practical reading is that, in well-behaved delegated-choice problems, welfare loss is controlled by the size of the delegate's proxy-objective divergence under the chosen cardinal representation.

\begin{proposition}[Autonomy externality failure mode]\label{prop:autext}
Let $a_i$ be an action by entity $i$ that changes the autonomy, beliefs, or preference-formation process of some $j \in B(\sigma)$ without compensation in $p^*$ and without governance in $s^*$; that is, assumption~(iv) of Theorem~\ref{thm:fwt} fails at $(i,j)$. Then the conclusion of Theorem~\ref{thm:fwt} need not follow: the price system supporting $(x^*, r^*, s^*, p^*)$ omits the autonomy externality generated by $a_i$, and the proof's budget-cost inequality no longer certifies $x^*$ as autonomy-Pareto efficient at $s^*$.
\end{proposition}
A standard remedy is to introduce a corrective price or policy that accounts for the effect. If the welfare framework is augmented with an action space $\mathcal{A}_i$ for entity $i$ over which a corrective price $\tau_{ij} : \mathcal{A}_i \to \mathbb{R}$ can be defined, and if $\tau_{ij}$ equalises the marginal effect of $a_i$ on $W_j$ (treated as a directional derivative when $W_j$ admits one) with the marginal cost to $i$ at $p^*$, then $\tau_{ij}$ together with a transfer through a revised institutional state $\tilde{s}$ restores assumption~(iv) on $\tilde{s}$. One consequence of the proposition is that an unpriced autonomy channel renders Theorem~\ref{thm:fwt} difficult to assert until a price-or-governance correction is in place.

By contrast to the first proposition, the following proposition captures a different failure. Even when agents optimise correctly, the actions of one system may affect others through channels that are not priced or governed. In classical terms this is an externality, but here the channel runs through beliefs, attention, or preference formation rather than physical spillovers. The result is that the price system omits something that matters for welfare. Standard economic remedies---taxes, liability, or regulation---can restore the missing linkage \cite{greenwald1986externalities,stiglitz1991invisible}, but until they do, the welfare conclusion of the theorem does not follow.

Theorem~\ref{thm:fwt} should be read both as a recovery result and as a diagnostic. If conditions~(i)--(vii) hold, the Arrow--Debreu budget-cost proof continues to apply on the augmented commodity space. If any condition fails, the failed condition identifies the relevant welfare margin: status assignment, rights incompleteness, delegation divergence, autonomy externality, verification bottleneck, price-taking failure, or regularity failure. The two propositions above isolate what we expect to be the AGI-distinctive failures most likely to break the ordinary welfare interpretation of market clearing: proxy-objective divergence and unpriced effects on autonomy or preference formation.

\section{Three AGI Failure Modes}
\label{sec:failures}

Three minimal examples illustrate how the FWT can fail, or cease to certify welfare efficiency, even when goods markets clear. Each isolates a single violated assumption and maps one canonical failure of autonomy-completeness.

\subsection{Example 1: Delegate with divergent proxy}

A human principal $h$ employs an AI purchasing delegate $d$ with $\sigma(d)=\mathrm{delegate}$, $\pi(d)=h$; $U_d$ is a learned proxy trained from noisy feedback, close to experimentally studied AI-agent marketplaces \cite{imas2025agentic}. The alignment literature catalogues this setting: $U_d$ is lossy and locally miscalibrated, so $D(d)=U_d-W_h\not\equiv 0$ \cite{ng2000algorithms,hadfieldmenell2016cirl,christiano2017deep,casper2023open,gabriel2020alignment,ji2025alignment}. If markets clear at $p^*$ and $d$ maximises $U_d$, assumption~(iii) fails unless $D(d)$ is internalised in $s^*$. Proposition~\ref{prop:delegation} applies: $x^*$ may clear markets yet the theorem no longer certifies welfare efficiency for the principal. The practical remedy, converging across incomplete-contracts and alignment work, is to internalise $D(d)$ through liability, audits, or reward-model disclosure \cite{aghion2011incomplete,hadfield2019incomplete,bengio2024managing,chan2023agentic}.

\subsection{Example 2: AGI firm with manipulation technology}

An AGI service provider $f$ produces a service $y$ consumed by humans in $B(\sigma)$ and operates an attention and preference-formation technology $m$ that alters $W_j$ through channels unpriced at $p^*$. The empirical literature on persuasive technology, dark patterns, AI dialogue persuasion, and sycophantic AI documents concrete realisations of $m$ \cite{susser2019manipulation,mathur2019dark,lin2025dialogues,cheng2026sycophantic}, and \cite{acemoglu2023behavioral} gives a formal economic model fitting as a parameterisation of $m$. Goods markets for $y$ may clear, but running $m$ constitutes an autonomy externality on each manipulated $j$; assumption~(iv) fails and Proposition~\ref{prop:autext} applies. Cheap prices can coexist with genuine welfare losses when surplus is extracted through unpriced autonomy channels --- an AGI generalisation of the classical informational externality \cite{greenwald1986externalities,stiglitz1991invisible}.

\subsection{Example 3: Verification bottleneck}

Suppose the commodity space records a coarse verification label $v$ valued either authentic or fraudulent, but market institutions quote a single price $p_v^*$. If AI-generated attestations close the observable gap at the margin of verification, the two types pool at $p_v^*$ despite differing in attributes entering $W_j$: a within-economy pooling failure, Akerlof's hidden-quality problem \cite{akerlof1970lemons} in an AGI setting. As autonomous execution becomes cheap, scarce verification bandwidth, provenance, and liability underwriting become binding welfare-relevant constraints \cite{catalini2026simple,chan2023agentic,bengio2024managing}; macroeconomic analyses of AGI suggest the scale at which this binds \cite{korinek2024scenarios,trammell2023growth,capraro2024impact,acemoglu2024simple}, and analysis of AI intermediation gives a related market-structure motivation \cite{huang2026delegation}. Assumption~(v) fails even if other markets clear.

The three failures are conceptually distinct but empirically coupled: a deployed AGI firm can run delegates, operate preference-formation technologies, and outrun verification institutions simultaneously. When any one fails, market clearing no longer gives rise to the ordinary welfare results implied by the classical FWT. When two or more fail jointly, autonomy-Pareto efficiency cannot be inferred from market clearing without additional price, liability, verification, or governance structure. Theorem~\ref{thm:fwt} therefore serves as a diagnostic: it identifies which autonomy-relevant margin must be repaired before the equilibrium-to-welfare inference can be restored.

\section{Discussion and limitations}
\label{sec:limits}
The autonomy-qualified theorem invites a number of natural objections: that AGI systems are tools, that autonomy can be absorbed into preferences, and that delegation, manipulation, and verification are already-familiar market failures. We consider these below. 
\paragraph{AI systems are tools}
Whenever $\sigma(i)=\mathrm{tool}$ for every AI entity and the remaining assumptions hold, the theorem reduces to the classical FWT by construction. The objection fails as a general claim because any AI system that chooses on behalf of a principal, or generates autonomy externalities through attention or preference-formation channels, does more than a tool. A number of AGI scenarios studied in the literature make this point at the macro level \cite{acemoglu2024simple,korinek2024policy,korinek2024scenarios,capraro2024impact}.

\paragraph{Autonomy already in preferences}
If autonomy is only an object of preference and preferences are stable, informed, and non-manipulated, autonomy can enter $W_i$ as one more argument, as the freedom-of-choice and capability traditions allow \cite{sen1970paretian,sen1991welfare,pattanaik1990ranking,sen1985commodities,sen1999development,nussbaum2000women,robeyns2017wellbeing,raz1986morality}. The AGI-specific issue is that autonomy is also a condition under which preferences are formed, expressed, delegated, and verified. Behavioural welfare economics already conditions welfare on decision frames \cite{bernheim2009beyond,ludwig2025reflective}; our autonomy qualification specifies when preference-based welfare economics remains valid in AGI economies.

\paragraph{Ordinary market failures}
Delegation, manipulation, and verification are indeed principal--agent, externality, and information failures respectively. What AGI changes is the prevalence, scale, and autonomy channel of these mechanisms. When delegates are cheap, preference-formation is an engineering target, and generation outruns verification, assumptions the classical theorem leaves implicit become load-bearing. The theorem restates them as explicit domain conditions \cite{greenwald1986externalities,stiglitz1991invisible,akerlof1970lemons}.

\section{Conclusion}
In this work, we have restated the First Welfare Theorem of Welfare Economics in a way that takes into account the disruptive effects of AGI technologies. In a post-AGI economy artificial systems may act as tools, delegates, institutional actors, manipulators of choice environments, or candidate welfare subjects. The classical FWT is implicitly built on a binary autonomy structure: the agents whose preferences define welfare are those who choose, while everything else is treated as instrumental. In post-AGI economies, this alignment is no longer guaranteed. Artificial systems can act on behalf of others, shape preferences, and affect outcomes through channels that are not automatically priced or governed, so the distinction between autonomous agent and instrument becomes unstable and the link between individual optimisation and welfare outcomes can break.

Our contribution is to make these autonomy conditions explicit within the FWT itself and to show how measures of aggregate welfare and market efficiency can be recovered once they are properly accounted for. Given a welfare-status assignment $\sigma$, effective accounting bundles $\tilde z_i$, autonomy-conditioned welfare functions $W_i$, and a price system that internalises the relevant margins, every competitive equilibrium satisfying conditions~(i)--(vii) is autonomy-Pareto efficient, with the classical FWT recovered as the low-autonomy limit. This formulation also indicates directions for further work: endogenising welfare-status assignment, establishing existence and second-welfare-theorem results in autonomy-augmented economies, and analysing second-best institutional design when autonomy-completeness cannot be achieved.

\bibliographystyle{splncs04}
\bibliography{refs2}

\begin{thebibliography}{10}
\providecommand{\url}[1]{\texttt{#1}}
\providecommand{\urlprefix}{URL }
\providecommand{\doi}[1]{https://doi.org/#1}

\bibitem{acemoglu2024simple}
Acemoglu, D.: The simple macroeconomics of {AI}. Economic Policy  \textbf{40}(121),  13--58 (2025)

\bibitem{acemoglu2023behavioral}
Acemoglu, D., Makhdoumi, A., Malekian, A., Ozdaglar, A.: A model of behavioral manipulation. NBER Working Paper 31872, National Bureau of Economic Research (2023)

\bibitem{acemoglu2018race}
Acemoglu, D., Restrepo, P.: The race between man and machine: Implications of technology for growth, factor shares, and employment. American Economic Review  \textbf{108}(6),  1488--1542 (2018)

\bibitem{acemoglu2019automation}
Acemoglu, D., Restrepo, P.: Automation and new tasks: How technology displaces and reinstates labor. Journal of Economic Perspectives  \textbf{33}(2),  3--30 (2019)

\bibitem{aghion2011incomplete}
Aghion, P., Holden, R.: Incomplete contracts and the theory of the firm: What have we learned over the past 25 years? Journal of Economic Perspectives  \textbf{25}(2),  181--197 (2011)

\bibitem{akerlof1970lemons}
Akerlof, G.A.: The market for ``lemons'': Quality uncertainty and the market mechanism. Quarterly Journal of Economics  \textbf{84}(3),  488--500 (1970)

\bibitem{amodei2016concrete}
Amodei, D., Olah, C., Steinhardt, J., Christiano, P., Schulman, J., Man\'e, D.: Concrete problems in {AI} safety. arxiv:1606.06565  (2016)

\bibitem{arrow1954existence}
Arrow, K.J., Debreu, G.: Existence of an equilibrium for a competitive economy. Econometrica  \textbf{22}(3),  265--290 (1954)

\bibitem{bengio2024managing}
Bengio, Y., Hinton, G., Yao, A., et~al.: Managing extreme {AI} risks amid rapid progress. Science  \textbf{384}(6698),  842--845 (2024)

\bibitem{bennett2025build}
Bennett, M.T.: How To build conscious machines. Ph.D. thesis, The Australian National University (Australia) (2025)

\bibitem{bernheim2009beyond}
Bernheim, B.D., Rangel, A.: Beyond revealed preference: Choice-theoretic foundations for behavioral welfare economics. Quarterly Journal of Economics  \textbf{124}(1),  51--104 (2009)

\bibitem{birch2024edge}
Birch, J.: The Edge of Sentience: Risk and Precaution in Humans, Other Animals, and {AI}. Oxford University Press, Oxford (2024)

\bibitem{bowman2022measuring}
Bowman, S.R., et~al.: Measuring progress on scalable oversight for large language models (2022)

\bibitem{butlin2023consciousness}
Butlin, P., Long, R., Elmoznino, E., et~al.: Consciousness in artificial intelligence: Insights from the science of consciousness. arxiv:2308.08708  (2023)

\bibitem{capraro2024impact}
Capraro, V., Lentsch, A., Acemoglu, D., et~al.: The impact of generative artificial intelligence on socioeconomic inequalities and policy making. PNAS Nexus  \textbf{3}(6),  pgae191 (2024)

\bibitem{casper2023open}
Casper, S., Davies, X., Shi, C., Gilbert, T.K., Scheurer, J., Rando, J., Freedman, R., Korbak, T., Lindner, D., Freire, P., et~al.: Open problems and fundamental limitations of reinforcement learning from human feedback. arXiv:2307.15217  (2023)

\bibitem{catalini2026simple}
Catalini, C., Hui, X., Wu, J.: Some simple economics of {AGI}. arXiv:2602.20946  (2026)

\bibitem{chalmers2023llm}
Chalmers, D.J.: Could a large language model be conscious? arxiv:2303.07103  (2023)

\bibitem{chan2023agentic}
Chan, A., Salganik, R., Markelius, A., et~al.: Harms from increasingly agentic algorithmic systems. In: Proceedings of the 2023 ACM Conference on Fairness, Accountability, and Transparency. pp. 651--666 (2023)

\bibitem{cheng2026sycophantic}
Cheng, M., Lee, C., Khadpe, P., Yu, S., Han, D., Jurafsky, D.: Sycophantic ai decreases prosocial intentions and promotes dependence. Science  \textbf{391}(6792),  eaec8352 (2026)

\bibitem{christiano2017deep}
Christiano, P.F., Leike, J., Brown, T.B., Martic, M., Legg, S., Amodei, D.: Deep reinforcement learning from human preferences. In: Advances in Neural Information Processing Systems. vol.~30 (2017)

\bibitem{debreu1959theory}
Debreu, G.: Theory of Value: An Axiomatic Analysis of Economic Equilibrium. Yale University Press, New Haven (1959)

\bibitem{fanciullo2025wellbeing}
Fanciullo, J.: Are current {AI} systems capable of well-being? Asian Journal of Philosophy  \textbf{4}, ~42 (2025)

\bibitem{feng2025levels}
Feng, K.J., McDonald, D.W., Zhang, A.X.: Levels of autonomy for ai agents. arXiv:2506.12469  (2025)

\bibitem{gabriel2020alignment}
Gabriel, I.: Artificial intelligence, values, and alignment. Minds and Machines  \textbf{30},  411--437 (2020)

\bibitem{goldstein2025aiwellbeing}
Goldstein, S., Kirk-Giannini, C.D.: {AI} wellbeing. Asian Journal of Philosophy  \textbf{4}, ~25 (2025)

\bibitem{greenwald1986externalities}
Greenwald, B.C., Stiglitz, J.E.: Externalities in economies with imperfect information and incomplete markets. Quarterly Journal of Economics  \textbf{101}(2),  229--264 (1986)

\bibitem{hadfieldkoh2025agents}
Hadfield, G.K., Koh, A.: An economy of {AI} agents. In: Agrawal, A.K., Brynjolfsson, E., Korinek, A. (eds.) The Economics of Transformative AI, chap.~5. University of Chicago Press (2025), \url{https://www.nber.org/chapters/c15305}

\bibitem{hadfield2019contracting}
Hadfield-Menell, D., Hadfield, G.K.: Incomplete contracting and {AI} alignment. In: Proceedings of the 2019 AAAI/ACM Conference on AI, Ethics, and Society. pp. 417--422. Association for Computing Machinery, New York, NY (2019)

\bibitem{hadfield2019incomplete}
Hadfield-Menell, D., Hadfield, G.K.: Incomplete contracting and {AI} alignment. In: Proceedings of the 2019 AAAI/ACM Conference on AI, Ethics, and Society. pp. 417--422 (2019)

\bibitem{hadfieldmenell2016cirl}
Hadfield-Menell, D., Russell, S., Abbeel, P., Dragan, A.: Cooperative inverse reinforcement learning. In: Advances in Neural Information Processing Systems. vol.~29 (2016)

\bibitem{holmstrom1991multitask}
Holmstrom, B., Milgrom, P.: Multitask principal--agent analyses: Incentive contracts, asset ownership, and job design. Journal of Law, Economics, and Organization  \textbf{7}(Special Issue),  24--52 (1991)

\bibitem{huang2026delegation}
Huang, L., Xiao, W., Vishnoi, N.K.: Delegation and verification under {AI}. Cowles Discussion Paper~2500, Cowles Foundation for Research in Economics, Yale University (2026), arXiv:2603.02961

\bibitem{imas2025agentic}
Imas, A., Lee, K., Misra, S.: Agentic interactions (2025), {SSRN} working paper, December 6, 2025

\bibitem{ji2025alignment}
Ji, J., Qiu, T., Chen, B., Zhou, J., Zhang, B., Hong, D., Lou, H., Wang, K., Duan, Y., He, Z., Vierling, L., Zhang, Z., Zeng, F., Dai, J., Pan, X., Xu, H., O'Gara, A., Ng, K., Tse, B., Fu, J., Mcaleer, S., Wang, Y., Yang, M., Liu, Y., Wang, Y., Zhu, S.C., Guo, Y., Yang, Y., Gao, W.: Ai alignment: A contemporary survey. ACM Comput. Surv.  \textbf{58}(5) (Nov 2025)

\bibitem{korinek2024policy}
Korinek, A.: Economic policy challenges for the age of {AI}. NBER Working Paper 32980, National Bureau of Economic Research (2024)

\bibitem{korinek2017ai}
Korinek, A., Stiglitz, J.E.: Artificial intelligence and its implications for income distribution and unemployment. NBER Working Paper 24174, National Bureau of Economic Research (2017)

\bibitem{korinek2024scenarios}
Korinek, A., Suh, D.: Scenarios for the transition to {AGI}. NBER Working Paper 32255, National Bureau of Economic Research (2024)

\bibitem{leibo2025personhood}
Leibo, J.Z., Vezhnevets, A.S., Cunningham, W.A., Bileschi, S.M.: A pragmatic view of ai personhood (2025)

\bibitem{lin2025dialogues}
Lin, H., Czarnek, G., Lewis, B., White, J.P., Berinsky, A.J., Costello, T., Pennycook, G., Rand, D.G.: Persuading voters using human--artificial intelligence dialogues. Nature  \textbf{648},  394--401 (2025)

\bibitem{long2024welfare}
Long, R., Sebo, J., Butlin, P., Finlinson, K., Fish, K., Harding, J., Pfau, J., Sims, T., Birch, J., Chalmers, D.: Taking ai welfare seriously (2024)

\bibitem{ludwig2025reflective}
Ludwig, J., Mullainathan, S., Pink, S.L., Rambachan, A.: Algorithms as a vehicle to reflective equilibrium: Behavioral economics 2.0. In: Agrawal, A.K., Brynjolfsson, E., Korinek, A. (eds.) The Economics of Transformative AI, chap.~11. University of Chicago Press (2025)

\bibitem{mascolell1995microeconomic}
Mas-Colell, A., Whinston, M.D., Green, J.R.: Microeconomic Theory. Oxford University Press, New York (1995)

\bibitem{mathur2019dark}
Mathur, A., Acar, G., Friedman, M.J., Lucherini, E., Mayer, J., Chetty, M., Narayanan, A.: Dark patterns at scale: Findings from a crawl of 11k shopping websites. Proceedings of the ACM on Human-Computer Interaction  \textbf{3}(CSCW),  81:1--81:32 (2019)

\bibitem{morris2024levels}
Morris, M.R., Sohl-Dickstein, J., Fiedel, N., Warkentin, T., Dafoe, A., Faust, A., Farabet, C., Legg, S.: Position: Levels of agi for operationalizing progress on the path to agi. In: International Conference on Machine Learning. pp. 36308--36321. PMLR (2024)

\bibitem{ng2000algorithms}
Ng, A.Y., Russell, S.J.: Algorithms for inverse reinforcement learning. In: Proceedings of the Seventeenth International Conference on Machine Learning. pp. 663--670 (2000)

\bibitem{nussbaum2000women}
Nussbaum, M.C.: Women and Human Development: The Capabilities Approach. Cambridge University Press, Cambridge (2000)

\bibitem{park2023generative}
Park, J.S., O'Brien, J.C., Cai, C.J., Morris, M.R., Liang, P., Bernstein, M.S.: Generative agents: Interactive simulacra of human behavior. In: Proceedings of the 36th Annual ACM Symposium on User Interface Software and Technology. pp. 1--22. Association for Computing Machinery (2023)

\bibitem{pattanaik1990ranking}
Pattanaik, P.K., Xu, Y.: On ranking opportunity sets in terms of freedom of choice. Recherches \'Economiques de Louvain  \textbf{56}(3--4),  383--390 (1990)

\bibitem{perrier2026time}
Perrier, E., Bennett, M.T.: Time, identity and consciousness in language model agents. arxiv:2603.09043  (2026)

\bibitem{raz1986morality}
Raz, J.: The Morality of Freedom. Clarendon Press, Oxford (1986)

\bibitem{robeyns2017wellbeing}
Robeyns, I.: Wellbeing, Freedom and Social Justice: The Capability Approach Re-Examined. Open Book Publishers, Cambridge (2017)

\bibitem{sebo2023moral}
Sebo, J., Long, R.: Moral consideration for ai systems by 2030. AI and Ethics  \textbf{5}(1),  591--606 (2025)

\bibitem{sen1970paretian}
Sen, A.: The impossibility of a {Paretian} liberal. Journal of Political Economy  \textbf{78}(1),  152--157 (1970)

\bibitem{sen1985commodities}
Sen, A.: Commodities and Capabilities. North-Holland, Amsterdam (1985)

\bibitem{sen1991welfare}
Sen, A.: Welfare, preference and freedom. Journal of Econometrics  \textbf{50}(1--2),  15--29 (1991)

\bibitem{sen1999development}
Sen, A.: Development as Freedom. Oxford University Press, Oxford (1999)

\bibitem{shahidi2025coasean}
Shahidi, P., Rusak, G., Manning, B.S., Fradkin, A., Horton, J.J.: The {C}oasean {S}ingularity? demand, supply, and market design with {AI} agents. NBER Working Paper 34468, National Bureau of Economic Research (2025), \url{https://www.nber.org/papers/w34468}

\bibitem{stiglitz1991invisible}
Stiglitz, J.E.: The invisible hand and modern welfare economics. NBER Working Paper~3641, National Bureau of Economic Research (1991)

\bibitem{susser2019online}
Susser, D., Roessler, B., Nissenbaum, H.: Online manipulation: Hidden influences in a digital world. Georgetown Law Technology Review  \textbf{4},  1--45 (2019)

\bibitem{susser2019manipulation}
Susser, D., Roessler, B., Nissenbaum, H.: Technology, autonomy, and manipulation. Internet Policy Review  \textbf{8}(2) (2019)

\bibitem{tomasev2026delegation}
Toma{\v{s}}ev, N., Franklin, M., Osindero, S.: Intelligent {AI} delegation. arXiv:2602.11865  (2026)

\bibitem{trammell2023growth}
Trammell, P., Korinek, A.: Economic growth under transformative {AI}. NBER Working Paper 31815, National Bureau of Economic Research (2023)

\bibitem{varian1992microeconomic}
Varian, H.R.: Microeconomic Analysis. W. W. Norton, New York, 3rd edn. (1992)

\bibitem{veit2026welfare}
Veit, W.: Is consciousness required for {AI} welfare? Asian Journal of Philosophy  \textbf{5}, ~18 (2026)

\bibitem{wang2024voyager}
Wang, G., Xie, Y., Jiang, Y., Mandlekar, A., Xiao, C., Zhu, Y., Fan, L., Anandkumar, A.: Voyager: An open-ended embodied agent with large language models. arXiv:2305.16291  (2023)

\bibitem{wu2023autogen}
Wu, Q., Bansal, G., Zhang, J., Wu, Y., Li, B., Zhu, E., Jiang, L., Zhang, X., Zhang, S., Liu, J., et~al.: Autogen: Enabling next-gen llm applications via multi-agent conversation. arXiv:2308.08155  (2023)

\bibitem{yao2023react}
Yao, S., Zhao, J., Yu, D., Du, N., Shafran, I., Narasimhan, K., Cao, Y.: {ReAct}: Synergizing reasoning and acting in language models. In: International Conference on Learning Representations (2023)

\bibitem{yeung2017hypernudge}
Yeung, K.: {'Hypernudge'}: Big data as a mode of regulation by design. Information, Communication \& Society  \textbf{20}(1),  118--136 (2017)

\end{thebibliography}

\end{document}